\documentclass[superscriptaddress,twocolumn]{revtex4}
\usepackage[utf8]{inputenc}
\usepackage{graphicx}
\usepackage{dcolumn}
\usepackage{bm}
\usepackage{multirow}
\usepackage{amssymb}
\usepackage{amsfonts}
\usepackage{amsmath}
\usepackage{amsthm}
\usepackage{tcolorbox}
\usepackage{color}
\usepackage{marvosym}
\newcommand{\mathbbm}[1]{\text{\usefont{U}{bbm}{m}{n}#1}}
\usepackage{tikz}
\usetikzlibrary{matrix,positioning,decorations.pathreplacing}

\newtheorem{theorem}{Theorem}

\newtheorem{lemma}[theorem]{Lemma}

\newcommand{\tr}{\text{tr}}

\newcommand{\one}{\mbox{$1 \hspace{-1.0mm}  {\bf l}$}}
\newcommand{\N}{\mathbb{N}}

\newcommand{\bra}[1]{\langle #1|}
\newcommand{\ket}[1]{|#1\rangle}

\newcommand{\be}{\begin{equation}}
\newcommand{\ee}{\end{equation}}
\newcommand{\bea}{\begin{eqnarray}}
\newcommand{\eea}{\end{eqnarray}}

\newcommand{\kommentar}[1]{}

\newcommand{\identity}{\mathbbm{1}}

\newcommand{\forget}[1]{}

\begin{document}

\title{Measurement outcomes that do not occur and their role in entanglement transformations}

\author{Martin Hebenstreit}
\affiliation{Institute for Theoretical Physics, University of Innsbruck, Technikerstr. 21A,  6020 Innsbruck, Austria}
\author{Matthias Englbrecht}
\affiliation{Institute for Theoretical Physics, University of Innsbruck, Technikerstr. 21A,  6020 Innsbruck, Austria}
\author{Cornelia Spee}
\affiliation{Naturwissenschaftlich-Technische  Fakultät,  Universität  Siegen,  Walter-Flex-Straße  3,  57068  Siegen,  Germany}
\author{Julio I. de Vicente}
\affiliation{Departamento de Matemáticas, Universidad Carlos III de Madrid, Avda. de la Universidad 30, E-28911, Leganés (Madrid), Spain}
\author{Barbara Kraus}
\affiliation{Institute for Theoretical Physics, University of Innsbruck, Technikerstr. 21A,  6020 Innsbruck, Austria}

\begin{abstract}
The characterization of transformations among entangled pure states via local operations assisted by classical communication (LOCC) is a crucial problem in quantum information theory for both theoretical and practical reasons. As LOCC has a highly intricate structure, sometimes the larger set of separable (SEP) maps is considered, which has a mathematically much simpler description. In the literature, mainly SEP maps consisting of invertible Kraus operators have been taken into account. In this paper we show that the consideration of those maps is not sufficient when deciding whether a state can be mapped to another via general SEP transformations. This is done by providing explicit examples of transformations among pure 3- and 5- qubits states, which are feasible via SEP maps containing singular Kraus operators, however, not possible via SEP maps containing solely regular Kraus operators. The key point that allows to construct the SEP maps is to introduce projective measurements that occur with probability zero on the input state. The fact that it is not sufficient to consider SEP maps composed out of regular Kraus operators even in the case of pure state transformations, also affects the results on LOCC transformations among pure states. However, we show that non-invertible Kraus operators do not help in state transformations under LOCC with finitely many rounds of classical communication, i.e. the necessary and sufficient condition for SEP transformations with invertible Kraus operators is still a necessary condition for convertibility under finite-round LOCC. Moreover, we show that the results on transformations via SEP that are not possible with LOCC (including infinitely many rounds of classical communication) presented in M. Hebenstreit, C. Spee, and B. Kraus, Phys. Rev. A \textbf{93}, 012339 (2016) are not affected.
\end{abstract}

\maketitle

\section{Introduction}
Understanding the entanglement properties of multipartite quantum systems plays a major role in both quantum information theory and condensed-matter physics. On the one hand, this allows to derive protocols for quantum communication such as secret sharing \cite{SecretSh} and schemes for quantum computation such as measurement-based computation \cite{RaBr01} to cite some examples. On the other hand, the entanglement structure of many-body systems can be used to characterize phase transitions \cite{AmFa08} and to devise schemes for numerical simulation using tensor network states \cite{orus}. In general, entanglement is considered to be one of the non-classical ingredients that allows quantum technologies to outperform their classical counterparts. For this reason, a resource theory of entanglement has been developed over the last two decades \cite{review}. This theory provides a rigorous framework that makes it possible to qualify and quantify this resource and to understand the fundamental possibilities and limitations behind its manipulation. However, many questions that have been long answered for bipartite systems turn out to be much more difficult when more parties are taken into account. Besides its fundamental interest, advancing further the resource theory of entanglement in the multipartite regime might lead to new genuinely many-body applications of quantum information theory.

Entanglement theory is formulated as a resource theory \cite{resource}. Such theories are built from the notion of the so-called free operations, which, due to the physical setting, are easily implementable and are therefore considered to be accessible at no cost. States that cannot be prepared with free operations acquire the status of a resource, in the sense that they might allow to overcome the limitations of what is possible by means of the free operations alone. Furthermore, the notion of free operations allows to define an operational partial order in the set of resource states: if there exists a free operation $\Lambda$ such that $\Lambda(\rho)=\sigma$, then $\rho$ is not less resourceful than $\sigma$. This is because any protocol that can be successfully implemented in this scenario (i.e.\ with free operations) starting from $\sigma$ can also be implemented successfully starting from $\rho$.  Functionals that preserve this ordering are considered to be resource quantifiers. Entanglement is a resource shared by different possibly space-separated parties. In this context, local operations assisted by  classical communication (LOCC) arise as a natural and operationally motivated choice of free operations. LOCC maps are built from local completely positive, trace preserving (CPTP) maps which the parties can correlate by exchanging classical communication. On the one hand, understanding LOCC allows to order and quantify the set of entangled states and to identify those that are potentially more useful. On the other hand, it provides protocols for the manipulation of this resource in practice.

A milestone result in this context is Nielsen's theorem \cite{nielsen}, which characterizes LOCC convertibility among pure bipartite states in terms of majorization. Unfortunately, the extension of Nielsen's theorem to the multipartite case is not straightforward at all. The mathematical characterization of the set of LOCC maps and LOCC transformations is extremely complicated due to the intricacies that arise when considering a potentially unbounded number of rounds of classical communication \cite{ChLe14}. Indeed, it is known that in contrast to bipartite pure state transformations \cite{LoPopescu}, no simplification can be placed on the number of rounds of classical communication that is sufficient to consider in general \cite{Ch11,ChHs17}.
Notwithstanding, several different works over the last years have led to considerable progress in our understanding of the rich entanglement structure of pure multipartite states. Reference \cite{Kr10} characterizes when pure multipartite qubit-states are related by local unitary (LU) transformations. Since LUs are invertible LOCC transformations, this defines equivalence classes of states with the same entanglement \cite{Gi02}. Reference \cite{slocc} introduces the notion of stochastic-LOCC (SLOCC) classes, which provides a coarse-grained classification of states with different entanglement properties. In more detail, two pure states are said to be in the same SLOCC class if they can be interconverted with non-vanishing probability by probabilistic LOCC. Thus, although this classification is based on an equivalence relation and, therefore, provides no sense of ordering, it tells us that LOCC manipulation can only occur within these classes. Indeed, LOCC convertibility has been later characterized within SLOCC classes with a simple mathematical structure such as the GHZ \cite{turgutghz} or the W \cite{turgutw} family. Another fruitful approach is to consider inner or outer approximations of the set of LOCC maps with a mathematically more tractable set of maps within a fixed SLOCC class. A natural and physically motivated inner approximation to LOCC is LOCC$_\mathbb{N}$, the set of LOCC maps implementable with a finite number of rounds of classical communication. The fact that such protocols have to terminate has allowed to characterize all states that are reachable by this class of transformations within a given (generic) SLOCC class and has allowed to identify multipartite protocols which cannot be boiled down to a concatenation of deterministic 1-round protocols as in the bipartite and the aforementioned multipartite case \cite{Spde17}. A particularly useful superset of LOCC is that of separable (SEP) maps, which are those CPTP maps that admit a Kraus decomposition in which all Kraus operators factorize in tensor products for each party \cite{rains}. Although it is known that the inclusion is strict, instances of protocols in which SEP outperforms LOCC are rare \cite{sepnotlocc} and, moreover, for certain tasks such as bipartite pure-state transformations they are known to be effectively the same \cite{gheorghiu}. In \cite{GoWa11} transformations among multipartite pure, fully entangled states (i.e.\ states for which the local density matrices are of full rank) within the same SLOCC class have been considered. There, a necessary and sufficient condition for the existence of a SEP map which transforms the pure initial to the pure final state has been provided.
However, there has been a constraint on the SEP map, which has been overlooked so far \cite{erratumGoWa11}. The criterion holds for SEP maps, whose Kraus operators are invertible. In the following we refer to this set of CPTP maps as SEP$_1$. Until now (with the exception of \cite{GoKr17,SaGo18}) SEP$_1$ has been considered as a superset of LOCC. The main reason why singular Kraus operators have not been considered (in the context of LOCC) is that they map the initial state into a state which is no longer in the same SLOCC class as the final state. However, the fact that the initial state could be annihilated by the Kraus operator has been ignored. Due to that, the condition on the existence of a SEP$_1$ map has been subsequently used to characterize LOCC convertibility among pure multipartite fully entangled states in several general systems such as 3-qubit states, 4-qubit states and 3-qutrit states \cite{deSp13,SaSc15,SpdV16,HeSp16}. In \cite{GoKr17,SaGo18}, however, it has been proven that generic pure fully entangled states, i.e. almost all fully entangled states, of more than three parties with arbitrary equal local dimension are isolated, i.e.\ they cannot be obtained from nor transformed to inequivalent pure fully entangled states by SEP and, hence, by LOCC.

In this work we explore the differences in what comes to fully-entangled pure-state transformations between SEP$_1$ and SEP and its consequences for deciding LOCC convertibility. Remarkably, we show that  necessary and sufficient conditions for SEP$_1$ convertibility are only sufficient for SEP convertibility. Note that this implies that SEP$_1$ is not necessarily a superset of LOCC. In order to prove this, we construct explicit examples of SEP transformations which are infeasible via $\textrm{SEP}_1$. Interestingly, these instances exist for systems of very small size and dimension such as 3-qubit and 5-qubit states. The crucial observation behind these constructions is that SEP transformations can contain, in contrast to $\textrm{SEP}_1$, projective Kraus operators which annihilate the initial state. Stated more operationally, since one can see that non-invertible Kraus operators which occur with non-zero probability do not need to be taken into account, the difference is given by measurement operators whose outcomes have zero probability when applied to the initial state. This does not only shed light on the role of the outcomes that cannot occur but, as explained above, it is important to decide how to interpret results that have been obtained previously based on the condition of \cite{GoWa11}. Importantly, we show here that for LOCC$_\N$ transformations among fully entangled states non-invertible Kraus operators do not need to be taken into account. In other words, the necessary and sufficient condition for SEP$_1$ convertibility remains a necessary condition for LOCC$_\mathbb{N}$ convertibility. Furthermore, we will also provide a general condition under which the conditions for the existence of a $\textrm{SEP}$ state transformation coincide with those for the existence of a $\textrm{SEP}_1$ map. This is used to show that the examples given in \cite{HeSp16} using the SEP$_1$ condition indeed provide pure state transformations which are possible via SEP but not via LOCC. On the other hand, the question of whether LOCC transformations in this context are only possible if they can be be implemented by $\textrm{SEP}_1$ remains unanswered, i.e.\ it is not clear whether the necessary and sufficient condition for SEP$_1$ convertibility is also a necessary condition for LOCC convertibility if one allows infinitely many rounds of classical communication. Figure \ref{fig1} summarizes the relation between aforementioned sets of pure state transformations incorporating the findings of this paper.

The structure of this paper is as follows. We will first define our notation. Then we will review the result of \cite{GoWa11} and provide necessary and sufficient conditions for transformations via SEP. We will then discuss the relations among the different separable classes of operations. In particular, we will provide examples for transformations that are only possible if singular Kraus operators are taken into account and we will show that $\textrm{LOCC}_\N$ transformations among fully entangled pure states are included in $\textrm{SEP}_1$. Moreover, we will derive a sufficient condition for SEP transformations to be implementable via $\textrm{SEP}_1$ and we will provide an adaptation of the proof for the examples of \cite{HeSp16}. Finally we will give a conclusion and an outlook.

\begin{figure}[h]
\vspace{10pt}
\begin{tcolorbox}
\textbf{Sets of operations for pure state transformations:}
\begin{itemize}
\item $\textrm{SEP}_1\subsetneq_T\textrm{SEP}$ [here]
\item $\textrm{LOCC}_\N\subseteq_T \textrm{SEP}_1$ [here]
\item In case the initial state possesses only the trivial local symmetry, $\identity$, then $\textrm{SEP}_1=_T\textrm{SEP}$ \cite{GoKr17,SaGo18}
\item $\textrm{LOCC}\subsetneq_T \textrm{SEP}$ [28, here]
\item $\textrm{SEP}_1 {\not\subseteq}_T \textrm{LOCC}$ [here]
\item $\textrm{LOCC}\overset{?}{\subset}_T\textrm{SEP}_1$
\item $\textrm{LOCC}_\N\overset{?}{\subseteq}_T\textrm{LOCC}$
\end{itemize}
\end{tcolorbox}
\vspace{10pt}

\centering
\includegraphics[width=0.5\textwidth]{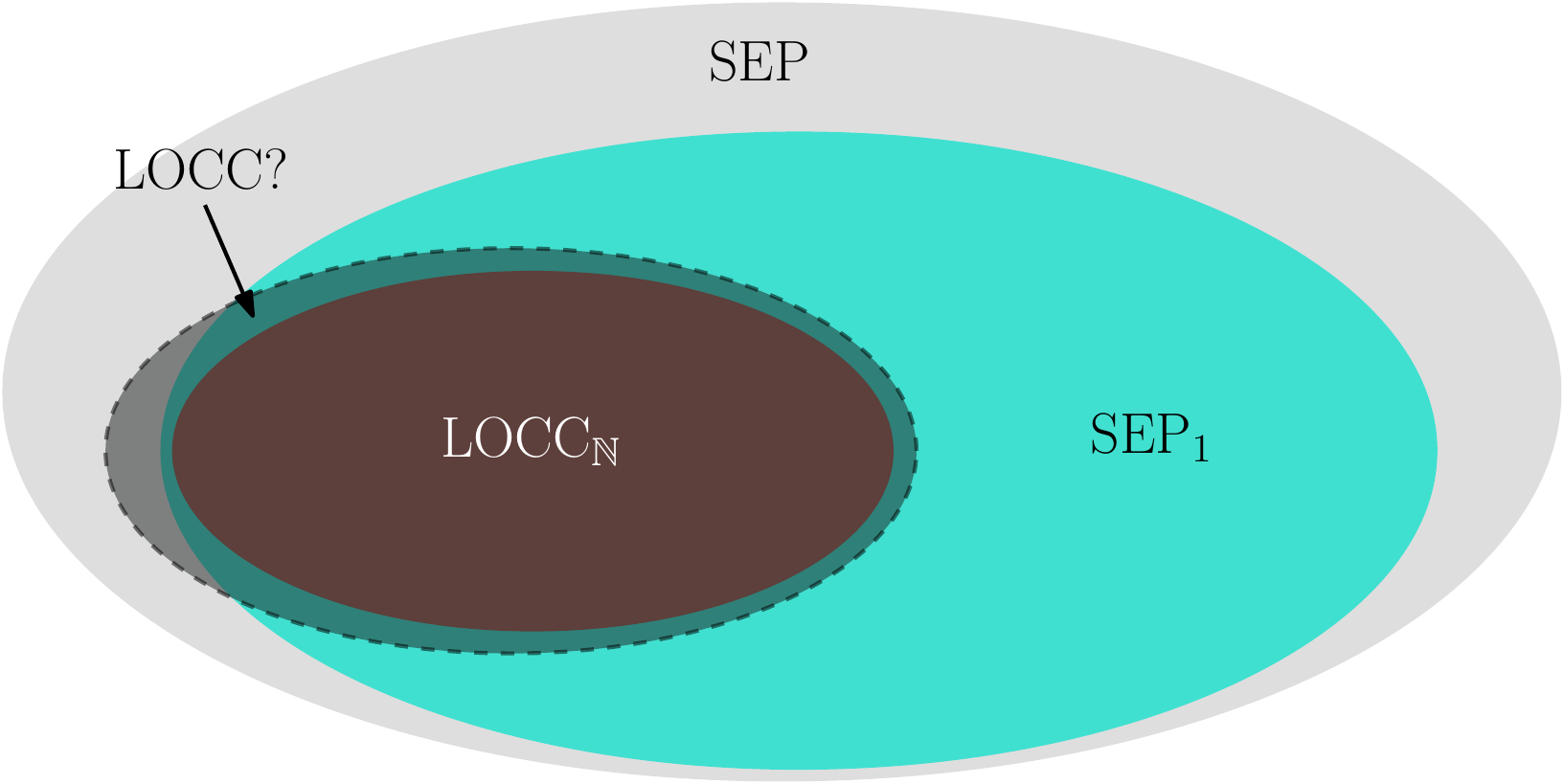}
\caption{Summary of the inclusion relations of possible pure state transformations among fully entangled states. As explained in Section \ref{sec:conventions}, we use the subscript $T$ to indicate that a relation should be understood in this sense.
 Surprisingly, there exist pure state transformations that are possible via SEP, but not via $\textrm{SEP}_1$, i.e., when considering only regular Kraus operators (see Section \ref{sec:examples}). Hence, $\textrm{SEP}_1\subsetneq_T\textrm{SEP}$.
It is an open problem whether there exist pure state transformations among fully entangled states that are possible via LOCC, but not via $\textrm{LOCC}_\N$. There do exist pure state transformations that are possible via $\textrm{SEP}$, but not via LOCC (which is proven here by adapting the proof of \cite{HeSp16}, in which SEP was considered to coincide with $\textrm{SEP}_1$). In fact, these examples show that $\textrm{SEP}_1 {\not\subseteq}_T \textrm{LOCC}$. It is currently not clear, whether there exist LOCC transformations which are not possible via $\textrm{SEP}_1$. In this article, we answer a related question by showing that finite round LOCC transformations among fully entanlged states are always possible via $\textrm{SEP}_1$, i.e., $\textrm{LOCC}_\N\subseteq_T \textrm{SEP}_1$.
\label{fig1}}
\end{figure}

 \section{Notation and preliminaries}
\label{sec:conventions}
In this work we consider pure states of an arbitrary number of parties $n$ and arbitrary local dimensions $\{d_i\}$ which are fully entangled. That is, states $|\psi\rangle$ in the Hilbert space $\mathcal{H}=\bigotimes_{i=1}^n\mathbb{C}^{d_i}$ such that the reduced density matrix for each party $i$, $\rho_i$, fulfills rank$\rho_i=d_i$. We call a state critical if $\rho_i \propto \identity$ for all parties $i$. We will consider transformations among fully entangled states, which are given by CPTP maps $\Lambda:\mathcal{B}(\mathcal{H})\to\mathcal{B}(\mathcal{H})$, where $\mathcal{B}(\mathcal{H})$ is the set of bounded linear operators acting on $H$ \footnote{Of course, one can also consider maps with different input and output Hilbert spaces but here we only analyze transformations among states with the same number of parties and local dimensions and we will take this into account in all forthcoming definitions.}. Every CPTP map admits a Kraus decomposition, i.e.\ it can be written as
\begin{equation}\label{kraus}
\Lambda(X)=\sum_i K_iXK_i^\dagger,
\end{equation}
for some set of operators $\{K_i\}\subset\mathcal{B}(\mathcal{H})$ fulfilling $\sum_iK_i^\dagger K_i=\one$ referred to as Kraus operators. As mentioned above, we will be interested in transformations among fully entangled pure states with particular subsets of the set of CPTP maps: SEP, SEP$_1$, LOCC$_\mathbb{N}$ and LOCC, whose definitions we provide in the following.

A CPTP map $\Lambda$ is said to be in SEP if it admits a Kraus decomposition with Kraus operators $\{K_i\}$ such that for all $i$ $K_i=\bigotimes_{j=1}^nK_i^{(j)}$ with $K_i^{(j)}\in M(d_j, \mathbb{C})$ $\forall j$, the last symbol referring to square matrices of size $d_j$ with complex entries. A SEP map $\Lambda$ is said to be in SEP$_1$ if there exists a Kraus decomposition of the above form which fulfills moreover that $K_i^{(j)}\in GL(d_j, \mathbb{C})$ $\forall i,j$, i.e.\ every Kraus operator is regular.

A CPTP map $\Lambda$ is said to be in LOCC$_\mathbb{N}$ with $m$ rounds of classical communication if it admits a Kraus decomposition with Kraus operators $\{K_i\}$ in which the index $i$ can be decomposed as a multi-index $i=(i_1\cdots i_m)$ so that
\begin{equation}
K_i=K_{(i_1\cdots i_m)}=\prod_{k=1}^m L_{i_k}(\{i_j\}_{j<k}),
\end{equation}
where the product should be understood from right to left (e.g. $\prod_{k=1}^2L_{i_k}=L_{i_2}L_{i_1}$) and
\begin{widetext}
\begin{equation}
L_{i_k}(\{i_j\}_{j<k})=U_{i_k}^{(1)}(\{i_j\}_{j<k})\otimes\cdots\otimes U_{i_k}^{(s_k-1)}(\{i_j\}_{j<k})\otimes P_{i_k}^{(s_k)}(\{i_j\}_{j<k})\otimes U_{i_k}^{(s_k+1)}(\{i_j\}_{j<k})\otimes\cdots\otimes U_{i_k}^{(n)}(\{i_j\}_{j<k}),\label{lio}
\end{equation}
\end{widetext}
where all the matrices labeled with $U$ are unitary, $s_k=s_k(\{i_j\}_{j<k})$ and
\begin{equation}
\sum_{i_k}(P_{i_k}^{(s_k)}(\{i_j\}_{j< k}))^\dag P_{i_k}^{(s_k)}(\{i_j\}_{j< k})=\one
\end{equation}
for all the possible values of $\{i_j\}_{j< k}$. That is, every element of the multi-index $i_k$ corresponds to a round, in which party $s_k$ implements a generalized measurement with measurement operators $\{P_{i_k}\}$. The identity of this party and the particular map he/she implements depend on all previous values of the elements of the multi-index $\{i_j\}_{j<k}$, which are known to every party through the use of classical communication. Then, party $s_k$ transmits to all other parties the precise outcome $i_k$ he/she obtains implementing the generalized measurement. Based on this value and all previous values of the elements of the multi-index, the remaining parties implement a unitary transformation to their share of the state, which concludes the round.

In order to define the set LOCC allowing for infinitely many rounds of classical communication, we first need to introduce the notion of composable LOCC$_\mathbb{N}$ maps. An LOCC$_{\mathbb{N}}$ map $\Lambda$ with $m$ rounds of classical communication and an LOCC$_{\mathbb{N}}$ map $\Lambda'$ with $m+1$ rounds of classical communication are said to be composable if they admit a Kraus decomposition as above with respective Kraus operators $\{K_i\}$ and $\{K'_i\}$ such that
\begin{equation}
K'_{(i_1\cdots i_{m+1})}=L_{i_{m+1}}(\{i_j\}_{j<m+1})K_{(i_1\cdots i_m)},
\end{equation}
where $L_{i_{m+1}}(\{i_j\}_{j<m+1})$ can be written as in Eq.\ (\ref{lio}). Thus, a CPTP map $\Lambda$ is said to be in LOCC if it is in $\textrm{LOCC}_\N$ or if it is the limit of a sequence of LOCC$_{\mathbb{N}}$ maps $\{\Lambda_i\}$ in which the maps $\Lambda_i$ and $\Lambda_{i+1}$ are composable $\forall i$.

Whenever there exists a map $\Lambda$ in SEP, SEP$_1$ or LOCC$_\mathbb{N}$ such that $\Lambda(|\psi\rangle\langle\psi|)=|\phi\rangle\langle\phi|$, we say that $|\psi\rangle$ can be converted into $|\phi\rangle$ by SEP, SEP$_1$ or LOCC$_\mathbb{N}$ operations. We make the analogous claim for LOCC with infinitely many rounds of classical communication if there exists a sequence of LOCC$_{\mathbb{N}}$ maps as above such that $\lim_{i\to\infty}||\Lambda_i(|\psi\rangle\langle\psi|)-|\phi\rangle\langle\phi|||=0$ in any matrix norm $||\cdot||$ of choice. From the above definitions it should be clear that $\textrm{SEP}_1\subset \textrm{SEP}$ and $\textrm{LOCC}_\mathbb{N}\subset \textrm{LOCC} \subset \textrm{SEP}$. However, we want to understand here whether the inclusion $X\subset Y$ translates into the existence of a transformation among fully entangled pure states by the operations given by $Y$ but not by the operations given by $X$ or whether both sets of transformation are equally powerful in this context. For this, we write $X=_TY$ if whenever there exists a map $\Lambda$ in $Y$ that transforms $|\psi\rangle$ into $|\phi\rangle$, there exists a map $\Lambda'$ in $X$ that transforms $|\psi\rangle$ into $|\phi\rangle$ and viceversa. On the other hand, we write $X\subsetneq_TY$ if $X\subset Y$ and for some states $|\psi\rangle$ and $|\phi\rangle$ the conversion $|\psi\rangle$ into $|\phi\rangle$ is possible within $Y$ but there exists no map in $X$ that transforms $|\psi\rangle$ into $|\phi\rangle$.

As explained in the introduction the considered transformations can only occur within SLOCC classes. Two states $|\psi\rangle,|\phi\rangle\in \mathcal{H}$ are said to be in the same SLOCC class if $|\phi\rangle=\bigotimes_{i=1}^ng_i|\psi\rangle$ with $g_i\in GL(d_i, \mathbb{C})$. We will consider for each SLOCC class a representative which we will refer to as $\ket{\psi}$. Other states in the SLOCC class are then identified by regular local operators acting on $\ket{\psi}$. Usually, we will use $ g \ket{\psi}$ with $g= \otimes_i g_i$  to denote the initial state and  $h \ket{\psi}$  with $h= \otimes_i h_i$ as the final state of a potential state transformation. Here, $g_i$ and $h_i$ are regular operators which reflects that we are interested in transformations among fully entangled states. Moreover, we will use the notation $G = g^\dagger g$ and $H = h^\dagger h$.

The stabilizer (or symmetry group) of $\ket{\psi}$, i.e., the set of local invertible operators leaving $\ket{\psi}$ invariant,  will be denoted by $\mathcal{S}_\psi$. More precisely, we have that
\begin{align}
\mathcal{S}_\psi = &\left\{ S : S\ket{\psi} = \ket{\psi}, S = S^{(1)} \otimes \ldots \otimes S^{(n)}, \right. \\ \nonumber
&\qquad \left. S^{(i)} \in GL(d_i, \mathbb{C}) \right\}.
\end{align}

Furthermore, we will denote by $\mathcal{N}_\psi$ the set of local operators which annihilate the state $\ket{\psi}$, i.e.

\begin{align}
\mathcal{N}_\psi = &\left\{ N : N\ket{\psi} = 0, N = N^{(1)} \otimes \ldots \otimes N^{(n)}, \right. \\ \nonumber
&\qquad \left. N^{(i)} \in M(d_i, \mathbb{C}) \right\}.
\end{align}

As we will see the stabilizer and the set annihilating the representative define which state transformations  are possible via SEP. In the next section we will discuss in detail the necessary and sufficient condition for such transformations, as well as the condition introduced previously in \cite{GoWa11}.

\section{State transformations}

In \cite{GoWa11} state transformation via separable maps  which  only involve regular matrices as Kraus operators have been considered.  In particular, the following necessary and sufficient condition for the existence of transformations among pure states via $\textrm{SEP}_1$ has been shown \cite{GoWa11}.
\begin{theorem}[\cite{GoWa11}]
The state $g \ket{\psi}$ can be transformed to $h \ket{\psi}$ via $\textrm{SEP}_1$ if and only if there exists a finite set of probabilities $p_k \geq 0$, $\sum_k p_k = 1$, and symmetries $\{S_k\}_k \subseteq \mathcal{S}_{\psi}$ such that

\begin{align}
\label{eq:sep}
\sum_{k} p_k S_k^\dagger H S_k = r G,
\end{align}
where $r = ||h\ket{\psi}||^2 / ||g\ket{\psi}||^2$.

\end{theorem}

It is currently unclear whether a pure state transformation that is possible via LOCC is always possible via $\textrm{SEP}_1$. However, we will show in the following that there exist state transformations via SEP which are impossible via $\textrm{SEP}_1$ and therefore SEP is strictly larger than $\textrm{SEP}_1$, i.e.\ $\textrm{SEP}_1\subsetneq_T \textrm{SEP}$. In order to see this, let us note that the Kraus operators occurring in a separable map might also annihilate the initial state, leading to more general maps. That is, operators $M_k$, with $M_k g \ket{\psi} = 0$ need to be taken into account. Hence, we have the following theorem characterizing SEP transformations.

\begin{theorem}
\label{ThSEP}
The state $g \ket{\psi}$ can be transformed to $h \ket{\psi}$ via SEP if and only if there exists a finite set of probabilities $p_k \geq 0$, $\sum_k p_k = 1$, symmetries $\{S_k\}_k \subseteq \mathcal{S}_\psi$, and local singular matrices $N_q \in \mathcal{N}_{g \psi}$ such that

\begin{align}
\label{eq:SEP}
 \frac{1}{r} \sum_{k} p_k S_k^\dagger H S_k  + g^\dagger \sum_q   N_q^\dagger N_q g =G,
\end{align}
where $r = ||h\ket{\psi}||^2 / ||g\ket{\psi}||^2$.

\end{theorem}

\begin{proof}

The proof of this theorem is analogous to the proof of Theorem 1 presented in \cite{GoWa11}. However, here non--invertible matrices have to be taken into account. We will first show that Eq. (\theequation) necessarily holds, if the transformation is possible via SEP.
Let $M_k$ ($N_q$) denote those Kraus operators, which reach the final state with non--vanishing probability (annihilate the initial state) respectively, i.e.
\bea M_k g\ket{\psi}/n_1&=&\sqrt{p_k} h\ket{\psi}/n_2,\\
N_q g\ket{\psi}&=&0 \eea

where $p_k > 0$ and $n_1=||g\ket{\psi}||$, $n_2=||h\ket{\psi}||$. Note that only finitely many measurement operators have to be taken into account (even if the stabilizer contains infinitely many elements) due to Caratheodory's theorem.
The first equation leads to $M_k=\sqrt{p_k} n_1/n_2 h S_k g^{-1}$ where $S_k$ is an element of the stabilizer $\mathcal{S}_\psi$ of $\ket{\psi}$. The completeness relation, $\sum_k M_k^\dagger M_k +\sum_q N_q^\dagger N_q=\one$ is hence equivalent to
\bea
\frac{1}{r} \sum_k p_k S_k^\dagger H S_k + g^\dagger \sum_q N_q^\dagger N_q g=G, \eea which proves that Eq. (\ref{eq:SEP}) has to be necessarily satisfied. That this condition is sufficient follows using the argument above in the reverse order.
\end{proof}

As we will see in the following there exist separable transformations which solely become possible when taking Kraus operators with vanishing probability into account. Note, however, that the results presented in \cite{GoKr17,SaGo18}, where it has been shown that almost all $n$-qudit states possess only the trivial stabilizer and are hence not convertible into any other state are not affected, as already proven in \cite{GoKr17}.

\section{Relations among classes of separable operations}

Whereas it is currently not clear whether pure state transformations that are possible via LOCC are always possible via $\textrm{SEP}_1$, we  will show here that $\textrm{SEP}_1$ does not coincide with SEP. Furthermore, we will show that a pure state transformation among fully entangled states that is possible via $\textrm{LOCC}_\N$ is necessarily possible via $\textrm{SEP}_1$ and therefore, any such pure state transformation via $\textrm{LOCC}_\N$ necessarily has to obey the conditions in Theorem 1. Moreover, we will derive sufficient conditions for which pure state transformations that are possible via SEP coincide with those via $\textrm{SEP}_1$. Finally, we will revisit the example presented in \cite{HeSp16} of  $\textrm{SEP}_1$ pure state transformations which cannot be realized via LOCC (taking infinitely many rounds into account). We show that the  statement remains true if one takes into account that there may be more transformations possible via SEP than $\textrm{SEP}_1$, implying that these are indeed examples of pure state transformations which can be achieved with SEP, however not with LOCC.

\subsection{Examples of SEP transformation that are not possible via $\textrm{SEP}_1$}
\label{sec:examples}
Let us start by presenting two distinguished examples of state transformations which are possible via SEP, but not via $\textrm{SEP}_1$. The first example is notable because the considered initial state has solely unitary stabilizer. The second example is found among three-qubit states and thus within the smallest possible multipartite quantum system.

Let us first consider the 5--qubit ring graph state, $\ket{\psi}$. A graph state is a special type of stabilizer state. For an introduction to graph states and stabilizer states we refer the reader to \cite{HeDu06}. The Pauli stabilizer of the state $\ket{\psi}$ is generated by $A_i=Z_{i-1} X_i Z_{i+1}$, for $1\leq i \leq 5$ and $Z_0=Z_5$, $Z_6=Z_1$. For this state we find that $\mathcal{S}_\psi =\left<\{A_i\}_{i=1}^5\right>$. To show this statement we use that if a critical state has finitely many unitary symmetries, then it has no other regular symmetries \cite{WaNo17} and that any graph state is a critical state. Considering the reduced density operators of three qubits it is straightforward to show that all local unitary symmetries of $\ket{\psi}$ are contained in its Pauli stabilizer and thus that there are only finitely many symmetries (for details see Appendix \ref{Appendix A}). We then consider the state transformation from $\ket{\psi}$ to a state $h\ket{\psi}$. Hence, we have that $G=\one$ and $h$ will be specified below. Using that the Pauli stabilizer is abelian, we obtain that Eq. (\ref{eq:sep}) is fulfilled only if $\tr(H P)=0$ for any non-trivial element $P$ of the Pauli stabilizer. Choosing $H=h^\dagger h= (1/2 \one +a Z)\otimes (1/2 \one +a X)\otimes (1/2 \one+a Z) \otimes 1/2 \one \otimes 1/2 \one$, for some $a\in (0,1/2)$, it holds that $\tr(H A_2)\neq 0$ and therefore the transformation  is not possible via $\textrm{SEP}_1$.

We construct now the SEP map which transforms $\ket{\psi}$ into $h \ket{\psi}$. In order to do so, we use the following projectors, which annihilate the initial state,
\begin{align}
Q_1&=\frac{1}{8}(\one+Z)\otimes (\one+ X)\otimes (\one- Z) \otimes \one^{\otimes 2}\label{eq:1}\\
Q_2&=\frac{1}{8}(\one+Z)\otimes (\one- X)\otimes (\one + Z) \otimes \one^{\otimes 2}\\
Q_3&=\frac{1}{8}(\one-Z)\otimes (\one+ X)\otimes (\one+ Z) \otimes \one^{\otimes 2}\\
Q_4&=\frac{1}{8}(\one-Z)\otimes (\one- X)\otimes (\one- Z) \otimes \one^{\otimes 2}.\label{eq:2}
\end{align}

The Kraus operators for the separable map are then given by:
$M_i=a_1 hQ_i$, for $i=1,2,3$ and with $a_1=2\sqrt{2a^3/((1/2+a)^2(1/2-a)(1/8+a^3))}$, $M_4 =2\sqrt{2a^3/((1/2-a)^3(1/8+a^3))} h Q_4$, and $M_5=\sqrt{1/(1/8+a^3)}h; M_6=M_5 A_1, M_7=M_5 A_3,M_8=M_5 A_1 A_3$. It is straightforward to verify the completeness relation $\sum_k M_k^\dagger M_k=\one$ and that the separable map corresponding to these Kraus operators indeed implements the transformation.

State transformations which are possible via SEP, but not via $\textrm{SEP}_1$, can be also found among three-qubit states. The following example is an adaption of the 5-qubit example presented above. Here, we consider a transformation from the 3-qubit ring graph state $\ket{\psi}$, which is LU equivalent to the 3-qubit GHZ state, to
$h_1 \otimes h_2 \otimes h_3 \ket{\psi}$. As before, $G=\one$ and we choose $h$ such that $H=h^\dagger h= (1/2 \one +a Z)\otimes (1/2 \one +a X)\otimes (1/2 \one+a Z)$, for some $a\in (0,1/2)$.
The stabilizer of the considered representative, $\ket{\psi}$, contains (by definition) the operators $A_i=Z_{i-1} X_i Z_{i+1}$, for $1\leq i \leq 3$ and $Z_0= Z_3$, $Z_4=Z_1$ and products thereof. However, in contrast to the five-qubit state above, the state considered here does possess additional symmetries, i.e. more than its Pauli stabilizer. Hence, in order to show that the considered transformation is not possible via $\textrm{SEP}_1$, we cannot use the same argument as above. However, we can resort to previous results on $\textrm{SEP}_1$ transformations among three-qubit states \cite{deSp13}, instead. In order to do so, we write the final state in the standard form  introduced in \cite{deSp13}. One obtains that the final state is up to local unitaries of the form
\begin{align}
h_x \otimes h_x \otimes h_x \ket{GHZ},
\end{align}
with $h_x^\dagger h_x\propto 1/2 \one +a X$. As has been shown in \cite{deSp13} it is not possible to reach states of the form above via  $\textrm{SEP}_1$.

Let us now show that the inclusion of singular Kraus operators allows to derive a map in SEP, which maps $\ket{\psi}$ into $h\ket{\psi}$. We use the projectors $Q_1',Q_2',Q_3',Q_4'$ defined such that $Q_j'\otimes \one^{\otimes 2}=Q_j$, for $Q_j$ as in Eqs. (\ref{eq:1}) to (\ref{eq:2}). Any of these operators annihilates the initial state. The Kraus operators for the separable map then take a similar form as in the previous example, namely:
$M_i=a_1 hQ_i'$, for $i=1,2,3$ and with $a_1=\sqrt{2a^3/((1/2+a)^2(1/2-a)(1/8+a^3))}$, $M_4 =\sqrt{2a^3/((1/2-a)^3(1/8+a^3))} h Q_4'$, and $M_5=(1/2) \sqrt{1/(1/8+a^3)}h; M_6=M_5 A_1, M_7=M_5 A_3,M_8=M_5 A_1 A_3$. Again it is straightforward to verify the completeness relation $\sum_k M_k^\dagger M_k=\one$ and that the separable map corresponding to these Kraus operators indeed implements the transformation. Hence, already for three qubits one can observe a difference among these sets of operations.
In the next section we will see that finite-round LOCC transformations among pure states are contained in $\textrm{SEP}_1$.

\subsection{State transformations using finitely many rounds of communication }
Finite-round LOCC protocols constitute a subset of LOCC that is of particular practical relevance. In this subsection we will show that there exists an $\textrm{LOCC}_\N$ transformation among fully entangled states only if Eq. (\ref{eq:sep}) holds, as stated in the following Lemma.

\begin{lemma}
\label{lemma:LOCCNinSEP1}
If there exists a map in $\textrm{LOCC}_\N$ which transforms a pure, fully entangled state into another, then there also exists a map in $\textrm{SEP}_1$ which accomplishes this transformation, i.e. $\textrm{LOCC}_\N\subseteq_T \textrm{SEP}_1$.
\end{lemma}
\begin{proof}
First, note that if an $\textrm{LOCC}_\N$ protocol is solely composed of measurements with regular measurement operators, then a Kraus decomposition of the map containing only local invertible Kraus operators exists. In this case, operators $N_q$ which contain singular matrices are thus not present in Eq. (\ref{eq:SEP}). Let us now show that the case, in which measurements including a singular measurement operator are performed, cannot occur. In order to see this, note that any local operator that annihilates a fully entangled state must be singular at not less than two sites. Any local operator that is singular at only one site, acting on a fully entangled state, thus yields a state with strictly positive norm, which, moreover, must have a rank deficient reduced density matrix at some site.
Let us now consider the first round in which one of the parties implements a measurement containing a singular measurement operator. Due to the considerations above, the resulting state corresponding to the singular measurement operator occurs with a strictly positive probability and, furthermore, the resulting state is no longer in the same SLOCC class as the final state. Hence, it is impossible to transform this state via LOCC into the final state \footnote{Note that it is impossible that this branch gets completely annihilated as any subsequent measurement (by any party) has to be complete. }. Hence, there is always a non--vanishing probability to obtain a state which is not in the same SLOCC class as the target state, which shows that it is impossible to deterministically transform one fully entangled state into another utilizing in any step of an $\textrm{LOCC}_\N$ protocol a singular matrix. This completes the proof.
\end{proof}

Whether the same holds true also when one includes the possibility of infinitely many rounds of classical communication is currently not clear. When dealing with infinite round protocols, many reasonings that apply to finite round protocols do not hold any more. Hence, the investigation of these protocols is more complicated.
In particular, the proof of Lemma \ref{lemma:LOCCNinSEP1} cannot be straightforwardly generalized to cover LOCC protocols with infinitely many rounds of classical communication. This is because such protocols could in principle implement a SEP transformation with non-invertible Kraus operators through a sequence of LOCC$_\mathbb{N}$ maps $\{\Lambda_i\}$ in which every $\Lambda_i$ has invertible Kraus operators. In fact, notice that if $|\psi\rangle$ cannot be transformed into $|\phi\rangle$ by SEP$_1$, this does not forbid the existence of a sequence of $\textrm{SEP}_1$ maps $\{\Lambda_i\}$ such that $||\Lambda_i(|\psi\rangle\langle\psi|)- |\phi\rangle\langle\phi|||\to 0$ as $i\to\infty$.

\subsection{States with unitary stabilizer}
In this section we will focus on state transformations within a SLOCC class for which a representative with solely unitary local symmetries can be found. It has been shown that whenever the stabilizer is finite, there always exists a representative for which the stabilizer is unitary, moreover, in case a critical state exists, it is the critical state \cite[Propositions 5 and 6]{GoWa11}.
We derive a necessary condition for SEP-transformations to be possible as well as a sufficient condition under which singular Kraus operators need not be taken into account, as stated in the following Lemma.

\begin{lemma}
\label{Lemma:unitary} Let $\mathcal{S}_\psi$ be unitary.
Consider an initial state $g \ket{\psi}$ and a final state $h \ket{\psi}$, where we choose w.l.o.g. $||g \ket{\psi}|| = ||h \ket{\psi}||=1$.
Then, if $g \ket{\psi}$ can be transformed into $h \ket{\psi}$ via SEP, it necessarily holds that $\tr (G)\geq \tr (H)$. Moreover, in case $\tr (G)= \tr (H)$, the SEP transformation is possible if and only if Eq. (\ref{eq:sep}) holds, i.e., if and only if a $\textrm{SEP}_1$ transformation is possible.

\end{lemma}
\begin{proof}
Consider Theorem \ref{ThSEP} for states with unitary stabilizers, i.e. $\mathcal{S}_\psi \subset U(d_1)\otimes \ldots \otimes U(d_n)$.
Taking the trace of Eq. (\ref{eq:SEP}) and using that $r=1$ we obtain
\bea  \tr(H) +p=\tr( G),\eea where $p=\tr( g^\dagger \sum_q N_q^\dagger N_q g)$. Note that $p\geq 0$. The assertion follows from the fact that $p=0$ iff $N_q=0$ $\forall q$, as the trace of positive operators is positive and as $g$ is regular.
\end{proof}

Hence, for normalized initial and final states, projective measurements need not be taken into account as long as $\tr(H)=\tr(G)$ (in case the stabilizer is unitary). Note that in the examples presented in Section \ref{sec:examples}, $\tr(H)=\tr(G)$ is obviously not fulfilled, when one normalizes the states.

In the following we will use Lemma \ref{Lemma:unitary} to provide an adaptation of the proof of the examples of pure state transformation that are possible via $\textrm{SEP}_1$ but not via LOCC given in \cite{HeSp16}.  In particular, we will take into account that there might exists  transformations that can be implemented via SEP but not via $\textrm{SEP}_1$.

\subsection{Examples of pure state transformations that are possible via SEP, but not via LOCC}
In \cite{HeSp16}, some of us have considered examples of pure state transformations that are possible via $\textrm{SEP}$, but not (infinite-round) LOCC. There, however, restricted LOCC operations have been considered, as $\textrm{SEP}_1$ was considered to be a superset of LOCC, instead of SEP. We will first briefly review the examples, as well as the main idea of the proof.
Then, we will present an adaptation of the proof to show that, indeed, these examples of state transformations are possible via SEP, but not LOCC. Let us mention here that these examples further show that $\textrm{SEP}_1 {\not\subseteq}_T \textrm{LOCC}$.

Let $\ket{\psi}$ denote the 3 qutrit seed states presented in \cite{HeSp16}. As shown in \cite{HeSp16} (see also \cite{BrLu04}), we have that $\mathcal{S}_\psi$ contains only (nine) unitary elements. We consider the transformation from $\ket{\psi}$ to $h \ket{\psi}$ (normalized), where $h=h_1\otimes h_2 \otimes \one$ as given in \cite{HeSp16}. The detailed definition of $h$ will not be relevant here. However, it will become important that $\tr(H)=\tr(\one)$ for $||h\ket{\psi}||=1$, as can be easily verified. As shown in \cite{HeSp16} $\ket{\psi}$ can be mapped to $h \ket{\psi}$ via $\textrm{SEP}_1$ (and therefore also via SEP). However, the proof that the transformation is not possible via LOCC has to be adapted. The reason for that is that in \cite{HeSp16}, we have argued that the transformation is not possible via LOCC as
\begin{enumerate}
\item $\ket{\psi}$ cannot be transformed to $h \ket{\psi}$ in a single round of classical communication (not even probabilistically) and
\item $\ket{\psi}$ is the only state that can be transformed to $h \ket{\psi}$ via $\textrm{SEP}_1$.
\end{enumerate}
However, statement 2. might no longer hold for SEP if one takes operators $N_q$ into account, i.e. if one considers the most general SEP operations.

Let us now assume that there exists an LOCC protocol transforming $\ket{\psi}$ into $h \ket{\psi}$ and show a contradiction. The LOCC protocol must be non-trivial, hence there must exist a first round, in which a non-trivial measurement is performed. The state at hand before this round is still (LU-equivalent to) $\ket{\psi}$ and all intermediate states afterwards are of the form $g_i  \ket{\psi}$, where $g_i$ acts trivially on all parties but $i$, and are normalized such that $||g_i  \ket{\psi}|| = 1$. It is important to note that for any such $g_i$, it can be shown that $\tr{G_i} = \tr{\identity} = \tr{H}$. As mentioned before, the last equality follows from the special form of the considered $h$ as  in \cite{HeSp16}. As the protocol must be deterministic, all intermediate states $g_i  \ket{\psi}$ must be convertible to $h  \ket{\psi}$ via LOCC and thus via SEP. As all the conditions for  Lemma \ref{Lemma:unitary} are satisfied, this lemma implies that $g_i$ and $h$ must satisfy Eq. (\ref{eq:sep}). However, in \cite{HeSp16} it is shown that the only state which fulfills (up to LU) this condition is $\ket{\psi}$ itself. Hence, all  $g_i  \ket{\psi}$ are LU-equivalent to $\ket{\psi}$. This contradicts the fact that we were considering a non-trivial round and proves that these transformations cannot be implemented via LOCC.

\section{Conclusion and outlook}
In this work we considered state transformations among pure fully entangled states via separable maps and certain subsets of SEP. In particular, we showed that for the most general transformation via SEP, it is essential to include Kraus operators that occur with zero probability when applied to the initial state as there exist state transformations which are not possible otherwise. This can already be observed in the three qubit  and five-qubit scenario. Moreover, we proved that finite-round LOCC protocols do neither require nor even allow for local measurements containing singular measurement operators in case the initial and the final state are fully entangled. In case the stabilizer is unitary we found a necessary condition for the existence of pure state transformations via SEP that is independent of the stabilizer. Moreover, we found constraints under which the existence of pure state transformations via SEP coincices with those via $\textrm{SEP}_1$. Latter we used to prove that the examples given in \cite{HeSp16} indeed correspond to pure state transformations which are possible via SEP and not via LOCC (including infinitely many rounds of classical communication).
The main open question is whether $\textrm{LOCC}\subseteq_T\textrm{SEP}_1$ holds or not. The answer to it would not only shed light on how  results of previous works need to be interpreted but in case it is negative, it would also show that there are pure state transformations which only become possible if infinite rounds of classical communication are utilized.

B.K. thanks R. Brieger and D. Sauerwein for discussions related to the characterization of the local unitary symmetries of special graph states.  M.E., M.H., and B.K. acknowledge financial support from the Austrian Science Fund (FWF) grant DK-ALM: W1259-N27 and the SFB BeyondC (Grant No. F7107-N38). Furthermore, M.E. and B.K. acknowledge support of the Austrian Academy of Sciences via the Innovation Fund ``Research, Science and Society''. CS acknowledges support by the Austrian Science Fund (FWF): J 4258-N27 and the ERC (Consolidator Grant 683107/TempoQ). J.I.dV. acknowledges financial support by the Spanish MINECO through grants MTM2017-84098-P and MTM2017-88385-P and by the Comunidad de Madrid through grant QUITEMAD-CM P2018/TCS­4342.

\appendix

\section{Local stabilizer of the $5$ qubit ring graph\label{Appendix A}}

Let $\ket{\psi}$ be the $5$-qubit ring graph state and let $T_\psi=\left<\{A_i\}_i\right>$ be its Pauli stabilizer. For an introduction to graph states see \cite{HeDu06}. We show here that $\mathcal{S}_\psi=T_\psi$. For a more general form of this proof see \cite{RaDa18}.

First we use that $\ket{\psi}$ is a connected graph state and thus a critical state. For critical states it holds that if the number of unitary elements in  $\mathcal{S}_\psi$ is finite, then these are the only elements of $\mathcal{S}_\psi$ \cite{WaNo17}. Hence, showing that any unitary element of $\mathcal{S}_\psi$ is an element of $T_\psi$ (which is a finite group) implies the statement.
In order to see that, note that for a graph states, $\ket{\psi}$ it holds that  $\rho \equiv \ket{\psi}\bra{\psi}\propto\sum_{T\in T_{\psi}}T$. Taking the partial trace over system 4,5 and over system 3,4, the condition

\bea
U\rho U^\dagger=\rho,
\label{eq:sym}
\eea
with $U=U_1\otimes \ldots \otimes U_5$ implies that

\bea
U_1ZU_1^\dagger\otimes U_2XU_2^\dagger\otimes U_3ZU_3^\dagger&=&Z\otimes X\otimes Z\\
U_1XU_1^\dagger\otimes U_2ZU_2^\dagger\otimes U_5ZU_5^\dagger&=&X\otimes Z\otimes Z
\eea
Hence, $U_1$ has to leave $X$ and $Z$ invariant under conjugation (up to a proportionality factor). It is straightforward to see that this implies that $U_1\in \left<X,Z\right>$ (up to a phase factor). A similar argument holds for any other $U_j$ ($j\neq 1$) due to the symmetry of the state. Next, we show that there exists no Pauli operator, $\sigma_1\otimes \ldots\otimes \sigma_5 \notin T_\psi$, which is a symmetry of the graph state. To demonstrate this, we note that the action of any Pauli operator on a graph state coincides with the action of an operator $Z^{\vec{k}}$, with $k_i\in \{0,1\}$  (up to some phase) (see e.g. \cite{HeDu06}), i.e.
\bea
U\ket{\psi}=e^{i\gamma}\sigma_1\otimes \ldots\otimes \sigma_5 \ket{\psi}=e^{i\tilde{\gamma}}Z^{\vec{k}}\ket{\psi}\overset{!}{=} \ket{\psi}.\label{aeq1}
\eea
For $\vec{k} \neq \vec{0}$ we have that $\ket{\psi}$ is orthogonal to $Z^{\vec{k}}\ket{\psi}$ (see e.g. \cite{HeDu06}) and thus for equation \ref{aeq1} to hold necessarily $U\in T_\psi$.
\end{document}